\documentclass[pdflatex,sn-basic]{sn-jnl}

\usepackage{graphicx}
\usepackage{amsmath,amssymb,amsfonts}
\usepackage{mathrsfs}
\usepackage[title]{appendix}
\usepackage[all,arc,2cell]{xy}
\UseAllTwocells

\hypersetup{
  bookmarksnumbered=true,
  colorlinks=true,
  linkcolor=blue,
  citecolor=blue,
  urlcolor=blue,
  pdftitle={A Topological Condition for Aggregation on Weak Preference Domains},
  pdfauthor={Sizhong Fang},
  pdfstartview=FitBH
}

\theoremstyle{thmstyleone}
\newtheorem{thm}{Theorem}[section]

\newtheorem{prop}[thm]{Proposition}
\newtheorem{lem}[thm]{Lemma}

\theoremstyle{thmstylethree}
\newtheorem{defn}[thm]{Definition}

\newtheorem{exmp}[thm]{Example}

\newtheorem{rem}[thm]{Remark}

\numberwithin{equation}{section}
\begin{document}

\title[Topological Aggregation on Weak Preference Domains]
{A Topological Condition for Aggregation on Weak Preference Domains}

\author*[1]{\fnm{Sizhong} \sur{Fang}}
\email{sizhong@uchicago.edu}

\affil[1]{\orgname{University of Chicago},
  \orgaddress{\city{Chicago}, \state{Illinois}, \country{USA}}}

\abstract{This paper provides a topological construction for studying social aggregation of weak preferences over a finite set of alternatives. We identify weak preferences with the cells of the type $A_{n-1}$ braid arrangement and, equivalently, with the simplices of the corresponding type $A_{n-1}$ Coxeter complex, which we call the preference complex. Then we define face-monotonicity by requiring a social welfare function to be order-preserving when viewed as a map between the preference-profile poset and the preference poset. We prove that if a common preference domain admits a unanimous, anonymous, and face-monotonic social welfare function, then every connected component of the geometric realization of its order complex is contractible. We then apply this construction to two restricted domains. For fixed-extrema domains, we prove that their order-complex realizations are contractible and define a unanimous, anonymous, and face-monotonic social welfare function using the meet operation. For single-peaked domains, we prove that the corresponding subcomplex of the preference complex has a contractible realization and, in the three-alternative case, construct a unanimous, anonymous, and face-monotonic social welfare function.}


\maketitle

\section{Introduction} 
Social choice theory studies how societies combine individual preferences over a set of choices into collective decisions. \cite{chichilnisky1980} introduced topology into social choice theory in 1980 and proved an impossibility theorem for continuous, unanimous, and anonymous aggregation, under the assumption that preference spaces are CW complexes. \cite{baryshnikov1993} provided a unified topological proof of both Arrow's and Chichilnisky's impossibility theorems. Later, \cite{baryshnikovRoot} extended the construction to prove the Gibbard-Satterthwaite Theorem. These impossibility theorems generally begin with unrestricted preference domains. However, in practice, this is rarely the case, motivating scholars to characterize restricted domains that avoid such impossibilities. \cite{ChichilniskyHeal} showed that contractibility is essentially a necessary and sufficient condition to resolve Chichilnisky's impossibility theorem, and later \cite{weinberger} strengthened their result.

This paper provides a bridge between this topological condition and weak preference domains over finite sets of alternatives. Whereas Baryshnikov and Root relate discrete preferences to topology using nerve complexes, we take a different approach by identifying weak preferences with the cells of the type $A_{n-1}$ braid arrangement, extending \cite{terao}'s identification of strict preferences with chambers of the type $A_{n-1}$ braid arrangement. The resulting cell poset is the face poset of the type $A_{n-1}$ Coxeter complex, which we call the \emph{preference complex}. This construction leads to two complementary views of the cell poset. The preference complex makes the restricted domain geometrically explicit: each preference corresponds to a simplex, and the face-order records the replacement of some strict comparisons with indifferences. In the order complex, each preference is a vertex, and each simplex is a chain. We will show that a face-monotonic social welfare function is an order-preserving poset map, which induces a continuous map on geometric realizations. When $P$ is a subcomplex of the preference complex, its preference-complex realization and order-complex realization are homeomorphic.

\medskip
\noindent\textbf{Theorem \ref{main theorem}.}
Let $\mathcal{A}$ be a finite set of alternatives and $\mathcal{N}$ a finite set of agents, with $m := |\mathcal{N}| \geq 2$. Let $P$ be a common domain of weak preferences over $\mathcal{A}$ that excludes the totally indifferent preference. If there exists a unanimous, anonymous, and face-monotonic social welfare function
\[
F \colon P^m \longrightarrow P,
\]
then every connected component of $|\Delta(P)|$ is contractible.
\medskip

This theorem provides a topological necessary condition for aggregation in finite settings. If one connected component of the realized preference domain is non-contractible, then no social welfare function satisfying all three axioms can exist. We will make a quick note on \emph{face-monotonicity} since it is not a conventional property. In Section \ref{section 3}, we show that preference domains and preference profiles can be viewed as posets ordered by refinement, and face-monotonicity means that when we view a social welfare function as a map between these posets, the map is order-preserving. Intuitively, face-monotonicity requires that when agents weaken some strict comparisons into indifferences, the social preference may only weaken accordingly. Also note that we only establish necessity; sufficiency in the finite setting remains open.

Section \ref{section 2} provides necessary background on braid arrangements, constructs the preference complex, and shows that the full weak preference domain on $n$ alternatives has a realization that is homeomorphic to $S^{n-2}$. Section~\ref{section 3} introduces unanimous, anonymous, and face-monotonic social welfare functions and proves the main necessary condition for their existence. Section~\ref{section 4} studies fixed-extrema domains, establishes their
contractibility, constructs aggregation rules on them, and presents an
application to housing allocation. Finally, Section~\ref{section 5} proves the
contractibility of the single-peaked domain and constructs an aggregation rule
for the case of three alternatives.

\section{Braid Arrangements and Weak Preferences} \label{section 2}

In this section, we first introduce the type $A_{n-1}$ braid arrangement and prove that there exists a one-to-one correspondence between the cells of the braid arrangement and weak preferences on a finite set of alternatives. We then show that the cell poset of the type $A_{n-1}$ braid arrangement can be realized as a simplicial complex in which each weak preference corresponds to a unique simplex. We call this simplicial complex the \emph{preference complex}.

Let $V = \{x \in \mathbb{R}^n \mid \sum_{i = 1}^n x_i = 0\}$. For $i < j$, we define the hyperplane $H_{ij} = \{x \in V \mid x_i = x_j\}$, where $x_i$ and $x_j$ are the $i$-th and $j$-th coordinates of $x$. Let $\mathcal{H}_{Br}=\{H_{ij}\mid 1\leq i<j\leq n\}$ be the collection of all such hyperplanes. We call $\mathcal{H}_{Br}$ the \emph{type $A_{n-1}$ braid arrangement}. A \emph{cell} of $\mathcal{H}_{Br}$ is a nonempty set defined by
\[
A = \bigcap_{1 \leq i<j \leq n} U_{ij},
\]
where $U_{ij}$ is either $H_{ij}$, $H_{ij}^+ := \{x \in V \mid x_i > x_j\}$, or $H_{ij}^- := \{x \in V \mid x_i < x_j\}$.

Every cell $A$ of the braid arrangement can be written as 
\[
A=\{x\in V \mid B_1>\cdots>B_k\}
\]
for some $1\leq k\leq n$, where each $B_i$ is a nonempty block of coordinates that are equal, and every coordinate in $B_i$ is greater than every coordinate in $B_{i+1}$. In other words, $\pi = (B_1 |\cdots |B_k)$ is an ordered partition of $\{x_1, \dots, x_n\}$.

By the definition of a cell, we know that all cells are pairwise disjoint, and each cell determines a unique \emph{sign sequence} $\sigma=(\sigma_{ij})_{i<j}$, where $\sigma_{ij}\in \{0,+,-\}$. For each $i<j$, we set $\sigma_{ij}=0$ if $U_{ij}=H_{ij}$, $\sigma_{ij}=+$ if $U_{ij}=H_{ij}^{+}$, and $\sigma_{ij}=-$ if $U_{ij}=H_{ij}^{-}$. The cells for which $\sigma_{ij} \neq 0$ for all $i < j$ are called \emph{chambers}. Thus, chambers are the cells whose corresponding ordered partitions have every indifference class as a singleton. We use $\Sigma$ to denote the set that contains all cells of the type $A_{n-1}$ braid arrangement. 

Given cells $A, B \in \Sigma$, we say $B$ is a \emph{face} of $A$, and we write $B \leq A$, if for each pair $(i, j)$ either $\sigma_{ij}(B) = 0$ or $\sigma_{ij} (B) = \sigma_{ij}(A)$. Intuitively, a proper face of a cell $A$ can be obtained by changing one or more inequalities $x_i < x_j$ or $x_i > x_j$ into equalities $x_i = x_j$, along with any further equalities required for consistency, or equivalently, repeatedly merging adjacent indifference classes of $A$'s corresponding ordered partition $\pi$.

\begin{defn}
We define the \emph{face relation} on $\Sigma$ by declaring that $B \leq A$ if and only if $B$ is a face of $A$. This relation makes $\Sigma$ into a poset.
\end{defn}

\begin{prop} \label{greatest lower bound} \cite[p.~29]{Abramenko&Brown}
Any two cells in $\Sigma$ have a greatest lower bound.
\end{prop}

We provide an example for the case $n=3$, corresponding to the $A_2$ braid arrangement. By definition, we know that $V = \{x \in \mathbb{R}^3 \mid \sum_{i = 1}^3 x_i = 0 \}$. Then one valid chamber $C$ is $C = \{x \in V \mid x_1 > x_2 >x_3\}$. For simplicity, we will denote it by $C_{123}$. The ordered partition corresponding to this chamber is $\pi_{C_{123}} = (\{x_1\}| \{x_2\}| \{x_3\})$. Then one face of $C_{123}$ is $\{x \in V \mid x_1 = x_2 > x_3\}$. We denote this set by $F_{\{12\}|3}$. The ordered partition corresponding to $F_{\{12\}|3}$ is $\pi_{F_{\{12\}|3}} = (\{x_1,x_2\}| \{x_3\})$, which is obtained by merging the adjacent blocks $\{{x_1}\}$ and $\{{x_2}\}$ in $\pi_{C_{123}}$. Figure \ref{A2 Braid Arrangement} below is a visualization of the cell decompositions of the $A_2$ braid arrangement.
\begin{figure} [htbp]
    \centering
    \includegraphics[width=0.5\linewidth]{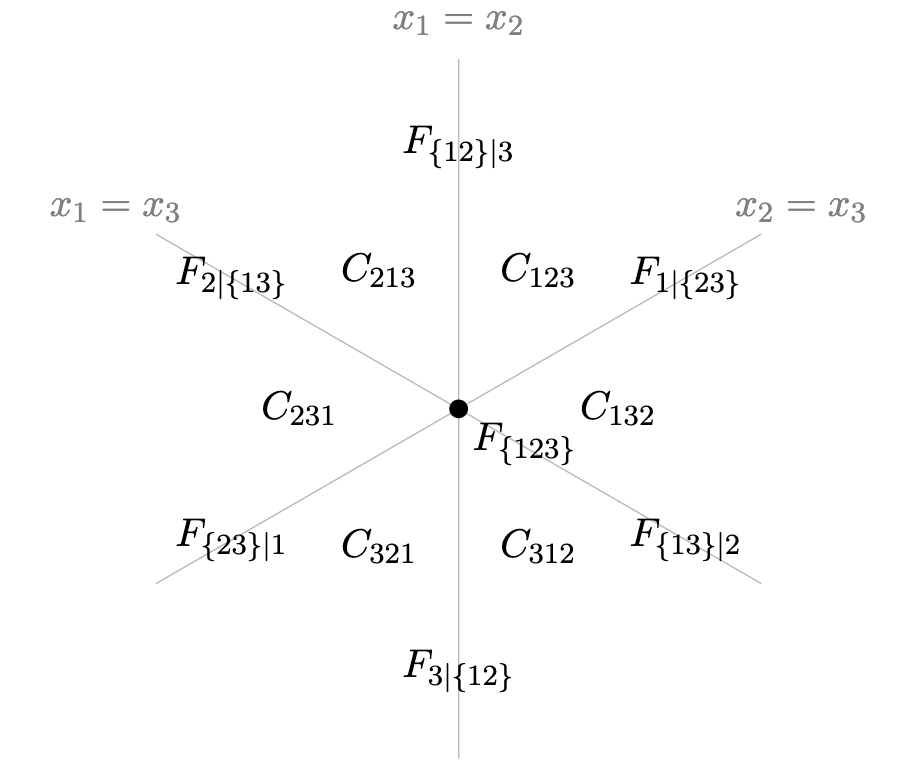}
    \caption{After restricting the arrangement to the essential part $V$, the $A_2$ braid arrangement consists of three hyperplanes in $\mathbb R^3$.}
    \label{A2 Braid Arrangement}
\end{figure}

Let $\mathcal{A} = \{a_1, \dots, a_n\}$ be the alternative set. Then every
weak preference over $\mathcal{A}$ can be written as an ordered partition of $\mathcal{A}$.
\[
\pi=(B_1|\cdots|B_k), \qquad B_1 \succ \cdots \succ B_k.
\]
Since each cell of the braid arrangement also has a unique corresponding ordered partition, we can relate these two constructions through the ordered partitions. The following theorem makes this part clear.

\begin{thm}[Weak preferences are cells of the braid arrangement] \label{cells and preference}
Let $\mathcal A=\{a_1,\dots,a_n\}$ be a finite set of alternatives, and let
$V=\left\{x\in \mathbb R^n \mid \sum_{i=1}^n x_i=0\right\}$. Then the cells of $\mathcal{H}_{Br}$ are in bijection with weak
preferences on $\mathcal A$. Under this bijection, a weak preference with
ordered indifference classes
\[
\pi=(B_1|\cdots|B_k),
\qquad
B_1\succ B_2\succ\cdots\succ B_k,
\]
corresponds to the cell
\[
C_\pi = \{\, x \in V \mid x_i = x_j \ \text{if } a_i, a_j \in B_r;\ 
    x_i > x_j \ \text{if } a_i \in B_r,\ a_j \in B_s,\ r < s \,\}.
\]
\end{thm}

\begin{proof}

First, we will show that, given a weak preference over $\mathcal{A}$ and its ordered partition $\pi=(B_1|\cdots|B_k)$, the corresponding cell $C_\pi$ is a valid cell of $\mathcal{H}_{Br}$. By our construction of $C_\pi$, it is sufficient to show that $C_\pi$ is nonempty. Let $\{r_1, \dots, r_k\}$ be an arbitrary finite decreasing sequence of positive integers.

We define $y_i = r_j$ if $a_i \in B_j$, where $i \in \{1, \dots, n\}$ and $j \in \{1, \dots, k\}$. Then $(y_1, \dots, y_n)$ satisfies the relationship defined by $C_\pi$. Now we define $x_i$ by
\[
x_i = y_i - \frac{1}{n}\sum_{j = 1}^{n} y_j.
\]

Since the $x_i$ are obtained by subtracting the same constant from the $y_i$, $(x_1,\dots,x_n)$ also satisfies the relationship defined by $C_\pi$. Furthermore, $\sum_{i=1}^{n} x_i = 0$, so $x = (x_1, \dots, x_n) \in C_\pi$. Thus, $C_\pi$ is nonempty, so it is a cell in $\mathcal{H}_{Br}$.

Conversely, let $C$ be a cell of the braid arrangement and choose
$x\in C$. Define a relation $\succeq_C$ on $\mathcal A$ by
\[
a_i\succeq_C a_j
\quad\Longleftrightarrow\quad
x_i\geq x_j.
\]
The sign of $x_i-x_j$ is the same for every point in $C$, so this relation is independent of the choice of $x\in C$. Since the usual order on $\mathbb R$ is complete and transitive, $\succeq_C$ is a weak preference relation.

It remains to verify that these constructions are mutually inverse. By definition, an ordered partition $\pi$ is completely determined by the pairwise relations of its elements, which correspond exactly to the sign sequence of the coordinates $x_i - x_j$. Since a cell $C$ in the braid arrangement is uniquely identified by this sign sequence, the mappings $\pi \mapsto C_\pi$ and $C \mapsto \pi_C$ preserve this underlying data. Consequently, $\pi_{C_\pi} = \pi$ and $C_{\pi_C} = C$, establishing the bijection.
\end{proof}

Recall that an \emph{abstract simplicial complex} with vertex set $V$ is a nonempty set $\Delta$ of finite subsets of $V$ satisfying the following two conditions: for every $v \in V$, $\{v\} \in  \Delta$; if $\delta \in \Delta$, then every subset of $\delta$ is also in $\Delta$. We call $\delta$ a \emph{simplex}, and a subset of $\delta$ a \emph{face} of $\delta$. The \emph{rank} of the simplex $\delta$ is defined as its cardinality $r$. The \emph{dimension} of $\delta$ is defined as $r - 1$. The singleton set $\{v\} \in \Delta$ has dimension $0$ and is called a \emph{vertex}. For our construction, we will also include the empty set $\varnothing$ with rank $0$ and dimension $-1$ in $\Delta$. A \emph{subcomplex} $\Delta'$ of $\Delta$ is a subset of $\Delta$ such that for every simplex $\delta \in \Delta'$, all the subsets of $\delta$ are in $\Delta'$. 

Given an abstract simplicial complex $\Delta$, its \emph{face poset}
$\mathcal{F}(\Delta)$ consists of the simplices of $\Delta$, including
$\varnothing$ under our convention, ordered by inclusion. Thus, the elements
of $\mathcal{F}(\Delta)$ are in one-to-one correspondence with the simplices
of $\Delta$, and the face poset determines $\Delta$ up to simplicial
isomorphism.

It is known that the cell poset $\Sigma$ of the type $A_{n-1}$ braid
arrangement is isomorphic to the face poset of the type $A_{n-1}$ Coxeter
complex. Through this isomorphism, each cell of the braid arrangement
corresponds uniquely to a simplex of the Coxeter complex. Accordingly, by a
slight abuse of notation, we also use $\Sigma$ to denote the resulting
abstract simplicial complex.

Since we proved there is a one-to-one correspondence between weak preferences and cells of $\Sigma$ in Theorem \ref{cells and preference}, we restate this correspondence when treating $\Sigma$ as a simplicial complex and conclude the following theorem.

\begin{thm} [Preference Complex] \label{(abstract) preference complex}
Let $\mathcal{A} = \{a_1, \dots, a_n \}$ be a finite set of alternatives. The set of weak preferences over $\mathcal{A}$ is in bijection with the set of simplices of the simplicial complex $\Sigma$ induced by the type $A_{n-1}$ braid arrangement. This simplicial complex is the type $A_{n-1}$ Coxeter complex, and we call $\Sigma$ the \emph{preference complex}.
\end{thm}

Given an abstract simplicial complex $\Delta$ with vertex set $V(\Delta)$, let
$\mathbb{R}^{V(\Delta)}$ denote the real vector space with basis
$\{e_v \mid v \in V(\Delta)\}$ indexed by the vertices of $\Delta$. For each
nonempty simplex $\delta \in \Delta$, we define its \emph{geometric simplex} $|\delta|$ as follows:
\[
|\delta|
=
\left\{
\sum_{v\in \delta} \lambda_v e_v
\ \middle|\
\lambda_v > 0 \text{ for all } v\in \delta,
\quad
\sum_{v\in \delta}\lambda_v=1
\right\}.
\]
The geometric realization of $\Delta$ is then defined as
\[
|\Delta| := \bigcup_{\varnothing \neq \delta\in \Delta} |\delta|.
\]
We call $|\Delta|$ the \emph{geometric realization} of $\Delta$. With the definition of geometric realization and Theorem \ref{(abstract) preference complex}, we obtain the following theorem

\begin{prop} \label{preference poset as preference complex}
Let $\Sigma$ be the preference complex of $\mathcal{A}=\{a_1,\dots,a_n\}$ and $|\Sigma|$ be its geometric realization. Then the weak preferences on
$\mathcal{A}$ other than the totally indifferent preference $\pi_0 = (\mathcal{A}) = \{a_1 \sim \cdots \sim a_n\}$ are in bijection with
the open geometric simplices of $|\Sigma|$, and these open simplices partition
$|\Sigma|$. Under this bijection, a preference with ordered indifference classes
$\pi=(B_1|\cdots|B_k)$, $k\ge 2$, corresponds to an open simplex of dimension
$k-2$. The totally indifferent preference $\pi_0=(\mathcal{A})$ corresponds to the
empty simplex $\varnothing$ and is therefore not represented in $|\Sigma|$.
\end{prop}

Now we move to the end result of this section. Given a finite-dimensional real vector space $V$ equipped with an inner product, we define the set $S(V) := \{x \in V \mid \lVert x \rVert = 1\}$. In other words, $S(V)$ is the unit sphere in $V$. Let $r:= \operatorname{dim} V$. We choose an orthonormal basis $\{e_1, \dots, e_r\}$ of $V$. Then the coordinate map $v = \sum_{i=1}^r x_i e_i \mapsto (x_1, \dots, x_r)$ induces a homeomorphism between $S(V)$ and $S^{r-1}$.

\begin{prop}[Geometric realization of the preference complex]
\label{prop:preference-complex-realization}
Let $\mathcal{A}=\{a_1,\ldots,a_n\}$, where $n\geq 2$, and let $\Sigma$
be the preference complex over $\mathcal{A}$, namely, the type $A_{n-1}$
Coxeter complex induced by the braid arrangement on
\[
V
=
\left\{
x\in\mathbb{R}^n
\;\middle|\;
\sum_{i=1}^n x_i=0
\right\}.
\]
There exists a homeomorphism
\[
\phi:|\Sigma|\longrightarrow S(V)
\]
such that, for every nonzero cell $A\in\Sigma$, the restriction
\[
\left.\phi\right|_{|A|}
:
|A|\longrightarrow A\cap S(V)
\]
is a homeomorphism. Consequently, the geometric realization of the
preference complex satisfies
\[
|\Sigma|
\cong
S(V)
\cong
S^{n-2},
\]
and is therefore noncontractible.
\end{prop}

\begin{proof}
The existence of $\phi$ and its restriction to each open simplex follows
from the realization theorem for the type $A_{n-1}$ Coxeter complex
\cite[p.~53]{Abramenko&Brown}. Since $V$ has dimension $n-1$, its unit
sphere is homeomorphic to $S^{n-2}$. Therefore,
\[
|\Sigma|\cong S(V)\cong S^{n-2}.
\]
\end{proof}

Figure \ref{fig:n=3 preference complex} provides a visualization of the geometric realization of the preference complex $|\Sigma|$ for the $n=3$ case. In this case, we have $|\Sigma| \cong S^{1}$. The vertices are the weak preferences with two indifference classes, and the 1-simplices are the six strict preferences. 

\begin{figure} [htbp]
    \centering
    \includegraphics[width=0.5\linewidth]{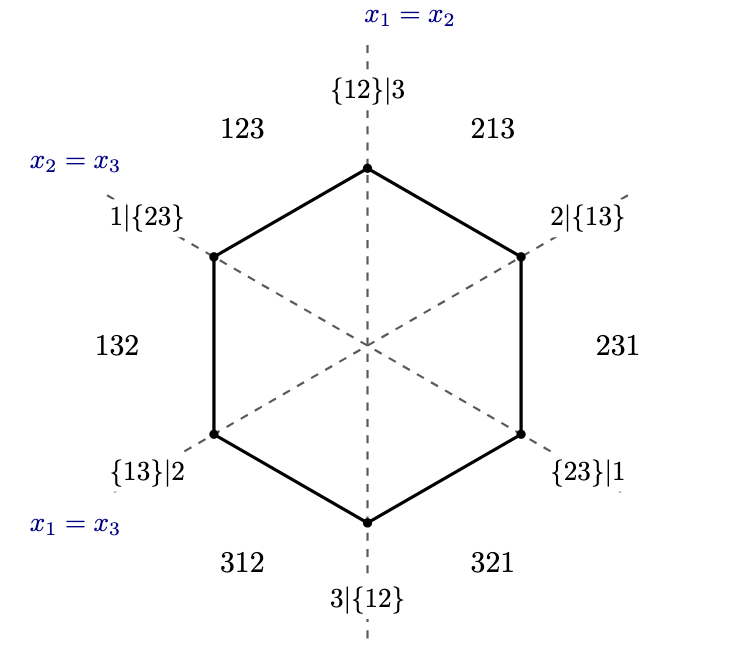}
 \caption{The geometric realization of the preference complex
$\lvert\Sigma\rvert$ for $n=3$. In this case,
$\lvert\Sigma\rvert\cong S^1$. For simplicity, $123$ denotes the
strict preference $a_1\succ a_2\succ a_3$, while
$\{13\}\vert 2$ denotes the weak preference
$a_1\sim a_3\succ a_2$.}
    \label{fig:n=3 preference complex}
\end{figure}

\section{A Topological Condition to Social Aggregation} \label{section 3}

In the previous section, we studied preferences over a finite set of alternatives for a single agent. In this section, we study a map that takes the agents’ preferences as input and outputs a single preference. These maps are called social welfare functions (SWFs). Readers can think of a social welfare function as a rule that considers everyone's preferences and comes up with a ``social preference''. For our purposes, we assume that the set of agents $\mathcal{N}$ and the set of alternatives $\mathcal{A}$ are both finite. We also assume that all agents have the same preference domain. There are many conditions that one can impose on a social welfare function, such as IIA and Pareto optimality, which appear in the famous Arrow's Impossibility Theorem. We study how the existence of a social welfare function satisfying unanimity, anonymity, and face-monotonicity is related to the contractibility of the geometric realization of the preference domain.

Throughout the paper, we have identified each weak preference with a cell of the type $A_{n-1}$ braid arrangement. We retain this identification in this section and write $A$ for a preference cell, rather than $\succeq$, to emphasize the face relation between preferences. 

\begin{rem} \label{rem: excluding empty cell}
By Theorem \ref{cells and preference} and Proposition \ref{preference poset as preference complex}, since each weak preference can be identified with a cell in $\Sigma$, the poset of cells of the type $A_{n-1}$ braid arrangement, we know that any preference domain $P$ is a subset of $\Sigma$. The preference domain $P$ can be viewed as a poset with the order inherited from $\Sigma$. Throughout this section we exclude the empty simplex $\varnothing$ (the totally indifferent preference) from every domain $P$ (including the full preference domain $\Sigma$): it is not represented in the geometric realization (Proposition \ref{preference poset as preference complex}), and since
it lies below every cell, keeping it would make every order complex $\Delta(P)$ a cone, hence contractible.
\end{rem}

\begin{defn} \label{preference profiles and social welfare functions}
Let $P$ be the preference domain, $\mathcal{N}$ be the set of agents, and $\mathcal{A}$ be the set of alternatives. After indexing the agents by $\{1, \dots, |\mathcal{N}|\}$, we define a \emph{preference profile} as a $|\mathcal{N}|$-tuple $(A_1, \dots, A_{|\mathcal{N}|})$, where $A_i \in P$ represents agent $i$'s preference over $\mathcal{A}$. Thus, the set of preference profiles is $P^{|\mathcal{N}|}$. A \emph{social welfare function} $F: P^{|\mathcal{N}|} \to P$ is a function that takes in a preference profile and outputs a preference.
\end{defn}

Note that in our case, all agents share the same preference domain $P$, which is also the codomain of $F$. Now we define the concepts of unanimity and anonymity. 

\begin{defn}
Let $X$ be a set and $N$ a finite index set. A map $f: X^{N} \to X$ is said to be \emph{unanimous} if $f(x, \dots, x) = x$, for every $x \in X$. Equivalently, $f$ is unanimous if $f \circ d = id$, where $d(x) = (x, \dots, x)$ is the diagonal map. A map $f: X^{N} \to X$ is said to be \emph{anonymous} if $f(x_1, \dots, x_N) = f(x_{\pi(1)}, \dots, x_{\pi(N)})$ for every permutation $\pi$ of $N$.
\end{defn}

Since $P$ is a poset that inherits the face order from $\Sigma$, we equip the set of preference profiles $P^{|\mathcal{N}|}$ with the coordinatewise face order and turn it into a poset:
\[
(A'_1,\ldots,A'_{|\mathcal{N}|})
\leq
(A_1,\ldots,A_{|\mathcal{N}|})
\]
if and only if
\[
A'_i  \leq  A_i
\qquad\text{for every }i\in\{1,\ldots,|\mathcal{N}|\}.
\]

A social welfare function $F$ is \emph{face-monotonic} if whenever $(A'_1,\ldots,A'_{|\mathcal{N}|}) \leq (A_1,\ldots,A_{|\mathcal{N}|})$, then \[F(A'_1,\ldots,A'_{|\mathcal{N}|}) \leq F(A_1,\ldots,A_{|\mathcal{N}|}).\] Equivalently, $F$ is an \emph{order-preserving} map between posets.

Intuitively, face-monotonicity means that if agents change some strict comparisons to indifferences, then the social preference can not introduce a new or reversed strict comparison. 

We now introduce the concept of order complex. 

\begin{defn}
Let $P$ be a poset. Define $\Delta(P)$ to be the abstract simplicial complex whose vertices are elements of $P$ ($\operatorname{Vert}(\Delta(P)) = P)$ and whose simplices are all finite chains of $P$, including the empty chain. We call $\Delta(P)$ the \emph{order complex} of $P$.
\end{defn}

\begin{rem}
Previously, we treated the poset $\Sigma$ as the face poset of a
simplicial complex, which we call the \emph{preference complex}, and used $\Sigma$ to denote this simplicial complex as well. However, it is important not to confuse the preference complex $\Sigma$ with its order complex $\Delta(\Sigma)$. For the preference complex $\Sigma$, every element of $\Sigma$ is a simplex with its corresponding dimension. For the order complex $\Delta(\Sigma)$, every element $A \in \Sigma$ is a vertex, and a k-simplex corresponds to a chain $A_0 < \cdots < A_k$ in $\Sigma$. For simplicity, we will continue to use $\Sigma$ for the preference complex and $\Delta(\Sigma)$ for its order complex. The results we obtained for $|\Sigma|$ and the geometric realization of its subcomplexes still hold for $|\Delta(\Sigma)|$ and the corresponding subcomplexes of $|\Delta(\Sigma)|$.
\end{rem}

\begin{prop} \label{homeo of geom of order complex and face complex}
The geometric realizations $|\Sigma|$ and $|\Delta(\Sigma)|$ are homeomorphic. More generally, for every subcomplex $\Gamma\subseteq\Sigma$, the order complex $\Delta(\Gamma)$ is a subcomplex of $\Delta(\Sigma)$. The realizations $|\Gamma|$ and $|\Delta(\Gamma)|$ are also homeomorphic.
\end{prop}

\begin{proof}
It is a known result that if a poset $P$ is the poset of nonempty faces of a simplicial complex $K$, then $\Delta(P)$ is the barycentric subdivision of $K$. Since a simplicial complex and its barycentric subdivision are homeomorphic, $|\Sigma| \cong |\Delta(\Sigma)|$. Furthermore, this homeomorphism restricts to a homeomorphism on subcomplex $\Gamma \subseteq \Sigma$, so $|\Gamma| \cong |\Delta(\Gamma)|$.
\end{proof}

To understand how the contractibility of the geometric realization of a restricted preference domain relates to the existence of a social welfare function that satisfies unanimity, anonymity, and face-monotonicity, we introduce a useful theorem.

\begin{thm} \label{W and C}
Let $X$ be a finite CW complex.
\begin{enumerate}
\item[(i)] If for some integer $m \ge 2$ there exists a continuous, unanimous, and anonymous map $f_m : X^m \to X$, then every connected component of $X$ is contractible \cite[Theorem~1.1]{weinberger}.
\item[(ii)] Conversely, if every connected component of $X$ is
contractible, then for every $m \ge 2$ there exists a
continuous, unanimous, and anonymous map $f_m : X^m \to X$
\cite[Theorem~1 and Corollary~2]{ChichilniskyHeal}.
\end{enumerate}
\end{thm}

We introduce one more lemma before applying Theorem \ref{W and C} to our construction.

Recall that $\mathbf{Pos}$ is the category whose objects are posets and whose morphisms are order-preserving maps, $\mathbf{Asc}$ is the category whose objects are abstract simplicial complexes and whose morphisms are simplicial maps, and $\mathbf{Top}$ is the category whose objects are topological spaces and whose morphisms are continuous maps. Furthermore, taking the order complex of a poset defines the functor $\Delta: \mathbf{Pos} \to \mathbf{Asc}$. Similarly, taking the geometric realization of an abstract simplicial complex defines the functor $|\cdot|: \mathbf{Asc} \to \mathbf{Top}$. The composition of the two functors is again a functor:
\[
|\Delta(\cdot)|: \mathbf{Pos} \longrightarrow \mathbf{Top}.
\]

\begin{lem} \label{delta N}
Let $m \ge 2$ be a finite positive integer and $\{\pi_i: P^m \to P\mid 1 \leq i \leq m\}$ be the canonical projections. Then there exists a homeomorphism $h:|\Delta(P^m)| \to |\Delta(P)|^m$, and 
\[
h = \prod_{i=1}^{m}|\Delta(\pi_i)|.
\]
\end{lem}

\begin{proof}
Barmak proved the $m=2$ case in \cite[Proposition~2.7.5, p.~31]{barmak}, when $P$ is a finite $T_0$ space. Since every finite poset can be identified with a finite $T_0$ space, the same argument can be applied and extended by induction. 
\end{proof}

\begin{thm} \label{main theorem}
Let $\mathcal{A}$ be a finite set of alternatives and $\mathcal{N}$ a finite set of agents, with $m := |\mathcal{N}| \geq 2$. Let $P$ be a common domain of weak preferences over $\mathcal{A}$ that excludes the totally indifferent preference. If there exists a unanimous, anonymous, and face-monotonic social welfare function
\[
F \colon P^m \longrightarrow P,
\]
then every connected component of $|\Delta(P)|$ is contractible.
\end{thm}

\begin{proof}

First, we show that if $F: P^m \to P$ is a face-monotonic social welfare function, then there exists a continuous map $f: |\Delta(P)|^m \to |\Delta(P)|$. Since $P^m$ and $P$ can be viewed as posets, and face-monotonicity is equivalent to being order-preserving, $F: P^m \to P$ is a morphism in $\mathbf{Pos}$. By Lemma \ref{delta N}, there exists a homeomorphism $h: |\Delta(P^m)| \to |\Delta(P)|^m$. Thus, the following diagram commutes:

\[
\xymatrix{
|\Delta(P^m)|
    \ar[r]^{|\Delta(F)|}
    \ar[d]_{h}
&
|\Delta(P)|
\\
|\Delta(P)|^m
    \ar[ur]_{f}
&
}
\]

where
\[
f = |\Delta(F)|\circ h^{-1}.
\]

Since $|\Delta(F)|$ and $h^{-1}$ are continuous, $f$ is continuous.

Next, we will show $f$ is unanimous. Define $d: P \to P^m$ to be the diagonal map between $P$ and $P^m$ and $\delta : |\Delta(P)| \to |\Delta(P)|^m$ to be the diagonal map between $|\Delta(P)|$ and $|\Delta(P)|^m$. We also define the projection maps $\lambda_i: |\Delta(P)|^m \to |\Delta(P)|$ for every $1 \leq i \leq m$. Then we know that $\lambda_i \circ \delta = \mathrm{id}_{|\Delta(P)|}$ for all $i$. We also have
\[
\lambda_i \circ (h \circ |\Delta(d)|) = |\Delta(\pi_i)| \circ |\Delta(d)| = |\Delta(\pi_i \circ d)| = |\Delta(\mathrm{id}_{P})| = \mathrm{id}_{|\Delta(P)|}.
\]
Thus, $\lambda_i \circ \delta = \lambda_i \circ (h \circ |\Delta(d)|)$, and $\delta = h \circ |\Delta(d)|$. Since $F$ is unanimous, $F \circ d = \mathrm{id}_{P}$. Then
\[
f \circ \delta = f \circ h \circ |\Delta(d)| = |\Delta(F)| \circ |\Delta(d)| = |\Delta(F \circ d)| = |\Delta(\mathrm{id}_{P})| = \mathrm{id}_{|\Delta(P)|},
\]
proving that $f$ is unanimous.

Third, we prove that $f$ is anonymous. For every $\sigma \in S_m$, let $\rho_\sigma (A_1, \dots, A_m) = (A_{\sigma(1)}, \dots, A_{\sigma(m)})$ be an automorphism of the poset $P^m$. This automorphism satisfies $\pi_i \circ \rho_\sigma = \pi_{\sigma(i)}$. Similarly, define $\tau_\sigma :|\Delta(P)|^m \to |\Delta(P)|^m$ to be the corresponding coordinate permutation, so $\lambda_i \circ \tau_\sigma = \lambda_{\sigma(i)}$. For each $i$:
\[
\lambda_i \circ (h \circ |\Delta(\rho_\sigma)|) = |\Delta(\pi_i)| \circ |\Delta(\rho_\sigma)| = |\Delta(\pi_{\sigma(i)})| = \lambda_{\sigma(i)} \circ h = \lambda_i \circ \tau_\sigma \circ h.
\]
Thus, we know that $h \circ |\Delta(\rho_\sigma)| = \tau_\sigma \circ h$. Equivalently, $h^{-1} \circ \tau_\sigma = |\Delta(\rho_\sigma)| \circ h^{-1}$. Since $F$ is anonymous, $F \circ \rho_\sigma = F$ for every $\sigma$. Hence, 
\[
f \circ \tau_\sigma = |\Delta(F)| \circ h^{-1} \circ \tau_\sigma = |\Delta(F)| \circ |\Delta(\rho_\sigma)| \circ h^{-1} = |\Delta(F \circ \rho_\sigma)| \circ h^{-1} = |\Delta(F)| \circ h^{-1} = f,
\]
proving that $f$ is anonymous.

From the above argument, we know that if there exists a unanimous, anonymous, and face-monotonic social welfare function $F: P^m \to P$, then there exists a unanimous, anonymous, and continuous map $f: |\Delta(P)|^m \to |\Delta(P)|$. Since $P$ is a finite poset, its geometric realization is a finite CW complex; thus, by Theorem \ref{W and C}, every connected component of $|\Delta(P)|$ is contractible.
\end{proof}

We now provide an example to show the significance of face-monotonicity.

\begin{exmp}
Consider the case of a committee of three people that uses Borda count to rank three proposals $a$, $b$, and $c$ with the two highest-ranked proposals advancing to a final vote. Each committee member assigns $2$, $1$, and $0$ points to their first-, second-, and third-ranked alternatives. Tied proposals receive the average of the points assigned to their corresponding positions. The case when proposals receive equal points is resolved according to the fixed priority
\[ a \mathrel{\triangleright} b \mathrel{\triangleright} c. \]

Consider the preference profile
\[ 
\mathbf{A} = \bigl( a\succ b\succ c,\; b\succ c\succ a,\; c\succ a\succ b \bigr). 
\] 

The corresponding Borda scores are 
\[ 
\begin{array}{c|ccc} & a & b & c \\ \hline A_1 & 2 & 1 & 0 \\ A_2 & 0 & 2 & 1 \\ A_3 & 1 & 0 & 2 \\ \hline \text{Total} & 3 & 3 & 3 \end{array} 
\] 
so the fixed priority rule gives 
\[
F(\mathbf{A})=a\succ b\succ c. 
\]

Now suppose that member $1$ weakens his reported preference to 
\[
A_{1}' = a \sim b \succ c,
\]
and the other two members’ preferences remain unchanged. The preference profile changes to 
\[ 
\mathbf{A}' = \bigl( a\sim b\succ c,\; b\succ c\succ a,\; c\succ a\succ b \bigr). 
\]
Because $A_1'\leq A_1$, we have $\mathbf{A}'\leq\mathbf{A}$. The new Borda scores are 
\[ 
\begin{array}{c|ccc} & a & b & c \\ \hline A_1' & 1.5 & 1.5 & 0 \\ A_2 & 0 & 2 & 1 \\ A_3 & 1 & 0 & 2 \\ \hline \text{Total} & 2.5 & 3.5 & 3 \end{array} 
\]
and hence \[ 
F(\mathbf{A}')=b\succ c\succ a. 
\]
However, $b\succ c\succ a \not\leq a\succ b\succ c$, so this Borda count rule is not face-monotonic.

Suppose that member $1$'s genuine preference is $a \sim b \succ c$. If member $1$ reports his genuine preference, then the social ranking is $b \succ c \succ a$, so proposals $b$ and $c$ advance to the final round. If member $1$ instead reports the strict refinement $a\succ b\succ c$, the social ranking becomes $a\succ b\succ c$, so $a$ and $b$ advance. Since member $1$ prefers the finalist set $\{a, b\}$ to $\{b, c\}$, the rule gives member $1$ an incentive to report a strict comparison between two proposals he genuinely regards as equal. 

Face-monotonicity prevents an aggregation rule from rewarding strategic polarization merely because a voter truthfully reports a moderate preference.
\end{exmp}

\section{Fixed-Extrema Subcomplex and Its Application} \label{section 4}

In this section, we consider the restricted preference domains in which each agent has a fixed most- or least-preferred alternative. We prove that the realizations of the fixed-extrema domains are contractible and construct a unanimous, anonymous, face-monotonic social welfare function for these domains. Then we apply the results on fixed-extrema preference domains to the housing allocation problem. 

\subsection{Contractibility of Fixed-Extrema Restricted Preference Domains} \label{section on extreme preference}

We assume the agent has an absolute least-preferred alternative throughout this subsection. Then by symmetry, we get the same result for the most-preferred case. Suppose $\mathcal{A} = \{a_1, \dots, a_n\}$ is a finite set of alternatives. Without loss of generality, we assume that $a_n$ is the absolute least-preferred alternative in $\mathcal{A}$, i.e., $a_i \succ a_n$ for all $a_i \in \mathcal{A} \setminus \{a_n\}$. 

We denote by $P^{\min}$ the poset of all weak preferences for which $a_n$ is the absolute least-preferred alternative, ordered by refinement as defined in the previous section. Then the order complex of $P^{\min}$, $\Delta(P^{\min})$, is the abstract simplicial complex that has all preferences in $P^{\min}$ as vertices and chains as simplices. We will show that $|\Delta(P^{\min})|$ is contractible by proving that $\Delta(P^{\min})$ is a cone. 

Recall that an abstract simplicial complex $K$ is a \emph{cone} with apex $v$ if and only if there exists a vertex $v$ such that for every simplex $\sigma \in K$, $\sigma \cup \{v\} \in K$. We denote by $A_0 := \{a_1 \sim \cdots \sim a_{n-1}\} \succ a_n$ the weak preference that the agent is indifferent among $\{a_1, \dots, a_{n-1}\}$, but strictly prefers each of them over $a_n$. Thus, $A_0$ corresponds to the cell $\{x \in V \mid x_1 = \cdots =x_{n-1} > x_n\}$. We show $\Delta(P^{\min})$ is a cone with apex $A_0$.

\begin{prop}
The abstract simplicial complex $\Delta(P^{\min})$ is a cone, and its geometric realization $|\Delta(P^{\min})|$ is contractible.
\end{prop}

\begin{proof}

Take an arbitrary simplex $\sigma \in \Delta(P^{\min})$. If
$A_0 \in \sigma$, then $\sigma \cup \{A_0\}=\sigma \in \Delta(P^{\min})$. Suppose that $A_0 \notin \sigma$. We want to show $\sigma \cup A_0 \in \Delta(P^{\min})$. Thus, it suffices to show $\sigma \cup \{A_0\}$ is a chain. Since every $F \in \sigma$ satisfies $x_i > x_n$ for all $1 \leq i \leq n-1$, and $A_0$ is not in $\sigma$, there must exist a pair $(j, k)$ such that $x_j > x_k > x_n$. $F$ is by definition more refined than $A_0$, so $A_0 < F$ for every $F \in \sigma$. Because $\sigma$ is already a chain and $A_0 < F$ for all $F \in \sigma$, $\sigma \cup A_0$ is a chain and a simplex in $\Delta(P^{\min})$. 

The geometric realization of a cone is again a cone, and a cone is known to be contractible, so $|\Delta(P^{\min})|$ is contractible.
\end{proof}

Assume that every agent has the same fixed-extrema domain $P^{\min}$ over a finite set of alternatives $\mathcal{A} = \{a_1 ,\dots, a_n\}$. We now define a social welfare function $F\colon (P^{\min})^{|\mathcal{N}|} \to P^{\min}$ that is unanimous, anonymous, and face-monotonic. Without loss of generality, we will again assume that $P^{\min}$ is the poset of all weak preferences for which $a_n$ is the absolute least-preferred alternative. Recall that the \emph{meet} $A \wedge B$ of two elements in a poset is their greatest lower bound. By Proposition \ref{greatest lower bound}, we know $P^{\min}$ is closed under meets.

\begin{prop} \label{swf for fixed-extrema}
Let $m=|\mathcal{N}|$, and suppose that every agent's preference domain is
$P^{\min}$, the poset of all weak preferences over $\mathcal{A}$ for which $a_n$
is the absolute least-preferred alternative; that is, $a_i\succ a_n$ for
every $1\leq i\leq n-1$. Define
\[
F\colon (P^{\min})^m\longrightarrow P^{\min},
\qquad
F(A_1,\ldots,A_m):=A_1\wedge\cdots\wedge A_m.
\]
Then $F$ is unanimous, anonymous, and face-monotonic.
\end{prop}

\begin{proof}
$F$ is unanimous and anonymous by construction. Suppose $(B_1, \dots, B_m) \leq (A_1, \dots, A_m)$. Then we know that $B_i \leq A_i$ for all $1\leq i \leq m$. Thus, $B_1 \wedge \cdots \wedge B_m \leq A_1\wedge\cdots\wedge A_m$, proving that $F$ is face-monotonic.
\end{proof}

\subsection{An Application to Housing Allocation}

In this subsection, we consider an application of our construction to the housing allocation problem. 

First, we assume that there are $n$ agents $\mathcal{N} = \{1, \dots, n\}$ and $n$ houses $\mathcal{O} = \{o_1, \dots, o_n\}$, so the cardinalities of the set of agents and the set of houses are the same. We define an allocation of houses as a surjection $\mu: \mathcal{N} \to \mathcal{O}$ such that $\mu(i) \neq \mu(j)$ for all $i \neq j$. Thus, when $|\mathcal{N}| = |\mathcal{O}|$, $\mu$ is a bijection. When there are $n$ agents and $n$ houses, we know that there are $n!$ allocations. We now take these allocations as the set of alternatives $\mathcal{A} = \{\mu\}$ and study the agent's preference over $\mathcal{A}$.  

Although the agent may formally have preferences over all allocations in $\mathcal A$, it is natural to assume that the agent cares only about the house assigned to them. Under this assumption, the agent's preferences over allocations are induced by their preferences over $\mathcal{O}$. The agent compares two allocations $\mu$ and $\nu$ by comparing the houses $\mu(i)$ and $\nu(i)$. This gives a natural restriction of the agent’s preference domain over $\mathcal A$. We study the contractibility of the geometric realization of this restricted domain.

Let $|\mathcal{N}| = |\mathcal{O}| = n$. We will fix an arbitrary agent $i$ throughout the analysis. From Proposition \ref{prop:preference-complex-realization}, we know that the geometric realization of the full preference complex $|\Sigma_{\mathcal{A}}|$ of agent $i$ over the housing allocations $\mathcal{A}$ is homeomorphic to $S^{n! - 2}$. Assume agent $i$ has full preferences over the houses $\mathcal{O} = \{o_1 \dots o_n\}$. We define 
\[
[o_j] := \{\mu \in \mathcal{A} \mid \mu(i) = o_j\},
\]
so each $[o_j]$ is the set of all allocations that assign $o_j$ to agent $i$. When agent $i$ cares only about their own assigned houses, agent $i$ is indifferent among all allocations in the same class $[o_j]$. Intuitively, $[o_j]$ can be viewed as an equivalence class. Thus, every preference over the housing allocations in the restricted domain is induced by a weak preference over $\{[o_1], \dots, [o_n]\}$. Then the canonical projection that maps $o_j$ to $[o_j]$ is a bijection between $\mathcal{O}$ and $\{[o_1], \dots, [o_n]\}$. Through the canonical projection, the restricted preference domain over $\mathcal{A}$ is in bijection with the set of all weak preferences over $\mathcal{O}$. 

At the cell and simplex levels, let $\Sigma_r \subset \Sigma_{\mathcal{A}}$ be the set of cells in $\Sigma_{\mathcal{A}}$ that correspond to the preferences in the restricted domain. Then each cell in $\Sigma_r$ corresponds to an ordered partition $\pi = ([D_1]| \cdots |[D_m])$ over $\mathcal{A}$, and each $[D_k]$ is a union of $[o_j]$. Since faces of cells are obtained by repeatedly merging adjacent blocks of
the corresponding ordered partitions, and a union of blocks that are each
unions of the $[o_j]$ is again a union of the $[o_j]$, every face of a cell
in $\Sigma_r$ lies in $\Sigma_r$, proving that $\Sigma_r$ is a subcomplex
of $\Sigma_{\mathcal{A}}$. 

Moreover, we can use the canonical projection to define a face-order-preserving bijection between the cells in $\Sigma_{\mathcal{O}}$ and cells in $\Sigma_r$ through their corresponding ordered partitions of ${\mathcal{O}}$ and ${\mathcal{A}}$. Thus, we have a well-defined simplicial isomorphism between $\Sigma_{\mathcal{O}}$ and $\Sigma_r$, which induces a homeomorphism between $|\Sigma_{\mathcal{O}}|$ and $|\Sigma_r|$. Because $|\Sigma_{\mathcal{O}}| \cong S^{n-2}$, $|\Sigma_r| \cong S^{n-2}$, so $|\Sigma_r|$ is not contractible. By Theorem \ref{main theorem}, when $n \ge 3$, there does not exist a social welfare function $F: (\Sigma_{r})^n \to \Sigma_{r}$ that is unanimous, anonymous, and face-monotonic. 

Now we consider the case when $|\mathcal{N}| > |\mathcal{O}|$. Then for every housing allocation $\mu$, there are $|\mathcal{N}| - |\mathcal{O}|$ agents who are not assigned houses. Then for each agent, there are $|\mathcal{O}| + 1$ potential cases for each allocation: the cases in which the agent is assigned to one of the houses in $\mathcal{O}$ and the case in which the agent is unmatched. Consider the case in which the agents only care about their own assignment under an allocation. In this case, each agent's restricted preference domain over $\mathcal{A}$ is induced by the agent's preference domain over $\mathcal{O} \cup \{\varnothing\}$, where $\varnothing$ means the agent is not assigned a house. Suppose that $\varnothing$ is the absolute least-preferred alternative in $\mathcal{O} \cup \{\varnothing\}$ for all agents. Then we know from Section \ref{section on extreme preference} that the geometric realization of each agent's restricted preference domain over $\mathcal{O} \cup \{\varnothing\}$ is contractible. Using an argument similar to that in the previous case, we know that when an arbitrary agent only cares about their own assignment in a housing allocation, and being unmatched is the least-preferred alternative, the geometric realization of the agent's restricted preference domain over $\mathcal{A}$ is contractible. However, the restricted preference domain over $\mathcal{A}$ induced from the agent's preferences over $\mathcal{O} \cup \{\varnothing\}$ does not have an absolute least-preferred housing allocation, so the map defined in Proposition \ref{swf for fixed-extrema} can not be applied here.

\section{The Single-Peaked Subcomplex and Its Contractibility} \label{section 5}

In this section, we study the restricted domain $P_{SP}$ of single-peaked preferences. Whereas the fixed-extrema domain was analyzed directly through its poset structure, here we study the cells corresponding to single-peaked preferences as a subcomplex $\Sigma_{SP}$ of the preference complex $\Sigma$. We show that $|\Sigma_{SP}|$, and hence $|\Delta(P_{SP})|$, is contractible. Finally, when $|\mathcal{A}|=3$, we construct, for every $|\mathcal{N}|\geq 2$, a social welfare function on $P_{SP}$ that is unanimous, anonymous, and face-monotonic.

Given a finite set of alternatives $\mathcal{A} = \{a_1, \dots, a_n\}$ with a strict linear ordering $\{a_1 < \cdots < a_n\}$, we say a subset $I \subseteq \mathcal{A}$ is an \emph{interval} if $I = [a_p, a_q] = \{a_p, a_{p+1}, \dots, a_{q-1}, a_q\}$ for some $p \leq q$. For our purposes, we also define $\varnothing$ as an interval.

\begin{defn}
A preference relation $\succeq$ is single-peaked if for every
alternative $a_m\in \mathcal{A}$, the strict upper level set
\[
U(a_m):=\{a_i\in \mathcal{A} \mid a_i\succ a_m\}
\]
is an interval with respect to the linear ordering
\[
a_1<a_2<\cdots<a_n.
\]
\end{defn}

We call $P_{SP}$ the set containing all single-peaked preferences and $\Sigma_{SP}$ the set containing all corresponding cells of the single-peaked preferences. Intuitively, a single-peaked preference has one most-preferred region, or ``peak,'' along the fixed ordering of alternatives. As one moves away from this peak in either direction, alternatives become weakly less preferred. Tax rates provide an example for which we can assume that people roughly have single-peaked preferences. In this case, the set of alternatives is the entire unit interval with the usual strict ordering of real numbers. Each person may have an ideal tax rate, or an ideal range of tax rates. As the tax rate moves away from this ideal point or range in either direction, the person's preference decreases because the alternative becomes further from what they consider optimal.

Since we know that the set of single-peaked preferences over $\mathcal{A}$ is a subset of the set of all preferences over $\mathcal{A}$ ($P_{SP} \subset P$), we ask a natural question: do we have a well-defined subcomplex of the preference complex, $\Sigma_{SP} \subset \Sigma$, such that this subcomplex corresponds to the single-peaked preferences? We will show that the answer to this question is affirmative and that the geometric realization of the single-peaked subcomplex $|\Sigma_{SP}|$ is contractible. Since $\Sigma_{SP}$ is a subcomplex of $\Sigma$, and by Proposition \ref{homeo of geom of order complex and face complex} $|\Sigma_{SP}| \cong |\Delta(P_{SP})|$, $|\Delta(P_{SP})|$ is contractible. We provide the geometric visualizations for the $n = 3$ and $n = 4$ cases. 

From Figures \ref{fig:single-peaked prefernce n = 3} and \ref{fig: single-peaked preferenc n = 4}, readers can verify that $\Sigma_{SP}$ is a well-defined subcomplex of $\Sigma$. Furthermore, as shown in Figure \ref{fig:single-peaked prefernce n = 3}, the geometric realization of $\Sigma_{SP}$ is a connected arc that is contractible.

\begin{figure} [htbp]
    \centering
    \includegraphics[width=1\linewidth]{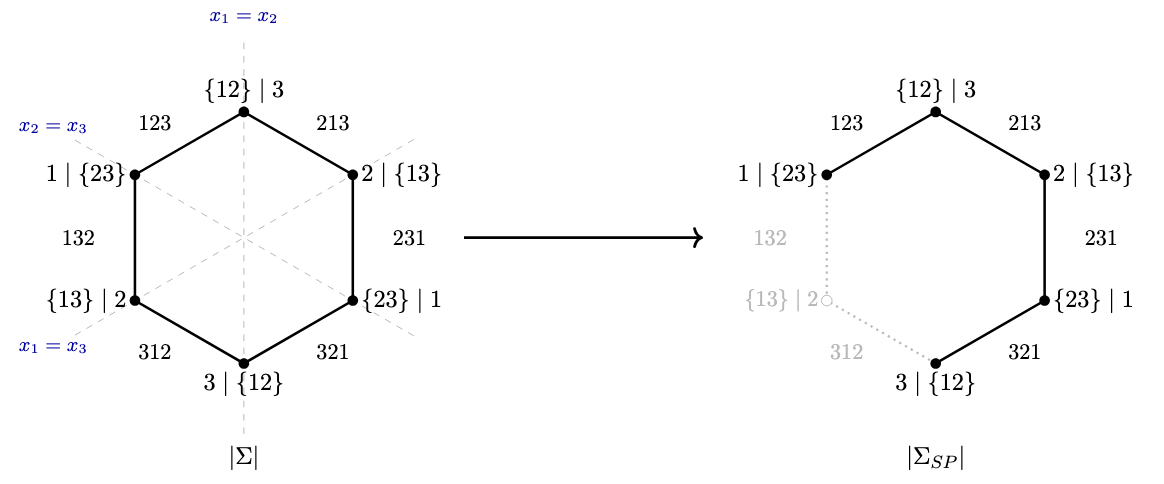}
    \caption{The single-peaked restriction for $n=3$: deleting the 1-simplices $(132)$ and $(312)$ together with the vertex $(\{13\}\mid 2)$ collapses $|\Sigma|$ to the contractible arc $|\Sigma_{SP}|$.}
    \label{fig:single-peaked prefernce n = 3}
\end{figure}

\begin{figure} [htbp]
    \centering
    \includegraphics[width=0.7\linewidth]{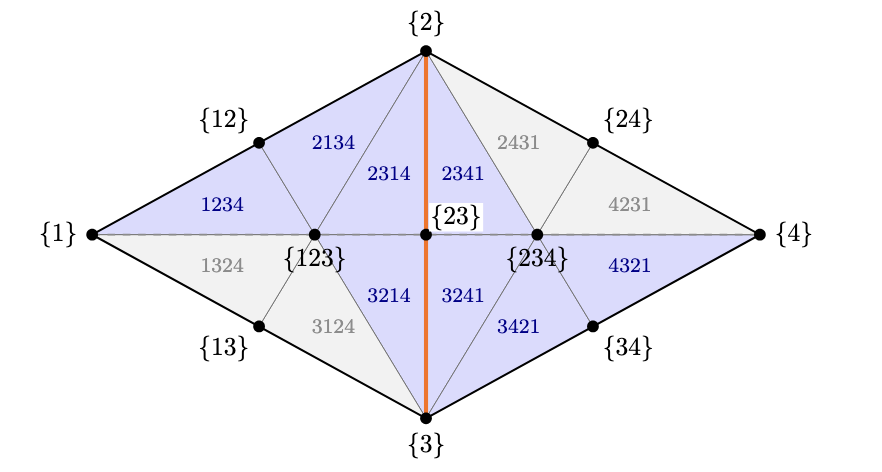}
    \caption{The single-peaked restriction for $n=4$: this consists of the two front faces of the boundary of a tetrahedron.}
    \label{fig: single-peaked preferenc n = 4}
\end{figure}

\begin{thm} \label{single-peakedness is contractible}
Let $|\mathcal{A}| \ge 3$. The geometric realization $|\Sigma_{SP}|$ of the single-peaked preference complex is contractible.
\end{thm}

The proof of this result is deferred to Appendix~\ref{app:sp-contractibility}.

We now construct a social welfare function on $P_{SP}$ that is unanimous, anonymous, and face-monotonic when $|\mathcal{A}| = 3$ and $|\mathcal{N}| \ge 2$. In this case, the nine single-peaked preferences are: 
\[
\begin{array}{lll}
p_1=a_1\succ(a_2\sim a_3),
&
p_2=a_1\succ a_2\succ a_3,
&
p_3=(a_1\sim a_2)\succ a_3,
\\[3pt]
p_4=a_2\succ a_1\succ a_3,
&
p_5=a_2\succ(a_1\sim a_3),
&
p_6=a_2\succ a_3\succ a_1,
\\[3pt]
p_7=(a_2\sim a_3)\succ a_1,
&
p_8=a_3\succ a_2\succ a_1,
&
p_9=a_3\succ(a_1\sim a_2).
\end{array}
\]
We equip $P_{SP}$ with the linear order 
\[ 
p_1\mathrel{\triangleleft}p_2\mathrel{\triangleleft}p_3 \mathrel{\triangleleft}p_4\mathrel{\triangleleft}p_5 \mathrel{\triangleleft}p_6\mathrel{\triangleleft}p_7 \mathrel{\triangleleft}p_8\mathrel{\triangleleft}p_9
\]
and define the coordinate map
\[ c\colon P_{SP}\longrightarrow\{1,\ldots,9\}, \qquad c(p_j)=j. \]
The face order on $P_{SP}$ is
\[ 
p_1<p_2>p_3<p_4>p_5<p_6>p_7<p_8>p_9. 
\]

\begin{prop}
Let $m = |\mathcal{N}|\geq 2$ and $|\mathcal{A}| = 3$. Define the social welfare function $F: (P_{SP})^m \to P_{SP}$
\[
F(A_1,\dots,A_m)
=
\begin{cases}
c^{-1}\!\left(
\operatorname{median}\bigl(c(A_1),\dots,c(A_m)\bigr)
\right),
& \text{if } m \text{ is odd},\\[6pt]
c^{-1}\!\left(
\operatorname{median}\bigl(c(A_1),\dots,c(A_m),5\bigr)
\right),
& \text{if } m \text{ is even}.
\end{cases}
\]
Then $F$ is unanimous, anonymous, and face-monotonic.
\end{prop}

\begin{proof}
First, the claim that $F$ is unanimous and anonymous follows straightforwardly from the construction. 

Suppose that a single agent $i$ replaces $A_i$ with a proper face
$A_i'<A_i$. Since total indifference is not in $P_{SP}$, $A_i$ must
be a strict preference. Therefore,
\[
c(A_i)=j\in\{2,4,6,8\},
\qquad
c(A_i')=j\pm1.
\]
Consider the odd-cardinality multiset whose median defines $F$. Replacing one
entry $j$ by the adjacent entry $j\pm1$ either leaves the median unchanged or
changes it from $j$ to $j\pm1$. Thus, the social preference either
remains unchanged or changes from $p_j$ to $p_{j\pm1}$. Since
\[
p_{j\pm1}<p_j
\]
in the face order, we have
\[
F(A_1,\ldots,A_i',\ldots,A_m)
\leq
F(A_1,\ldots,A_i,\ldots,A_m).
\]
For the general case, assume that $(B_1, \dots, B_m) \leq (A_1, \dots, A_m)$. We introduce the intermediate profiles:
\[
\begin{aligned}
\mathbf{A}^{(0)}
&= (A_1,A_2,\ldots,A_m),\\
\mathbf{A}^{(1)}
&= (B_1,A_2,\ldots,A_m),\\
\mathbf{A}^{(2)}
&= (B_1,B_2,A_3,\ldots,A_m),\\
&\ \vdots\\
\mathbf{A}^{(m)}
&= (B_1,\ldots,B_m).
\end{aligned}
\]
Then by the previous argument, for all $1 \leq k \leq m$
\[
F(\mathbf{A}^{(k)}) \leq F(\mathbf{A}^{(k-1)}). 
\]
Thus, by transitivity of the order, $F(B_1, \dots, B_m) \leq F(A_1, \dots, A_m)$. $F$ is face-monotonic.
\end{proof}

Constructing a unanimous, anonymous, and face-monotonic social welfare
function becomes more complicated when $|\mathcal{A}|>3$, and it remains
unclear whether such a function always exists.

\section*{Acknowledgement}
The author thanks Isabelle Steinmann for her guidance through the process. The author thanks Joseph Root and Shmuel Weinberger for helpful discussions and valuable suggestions for this project. The author thanks Peter May and the University of Chicago Mathematics Research Experience for Undergraduates Program for supporting this project. The author thanks Xia Wu, Wei Fang, and Siwen Fang for their constant support. 

\section*{Statements and Declarations}

\noindent\textbf{Funding}\quad
This work was supported by a stipend from the University of Chicago Mathematics 2026 Research Experience for Undergraduates Program

\noindent\textbf{Competing interests}\quad
The author has no relevant financial or non-financial interests to disclose.

\noindent\textbf{Data availability}\quad
No datasets were generated or analyzed during the current study.

\begin{appendices}
\section{Proof of Theorem \ref{single-peakedness is contractible}} \label{app:sp-contractibility}
We provide the proof of Theorem \ref{single-peakedness is contractible} here. First, we show $\Sigma_{SP}$ is a subcomplex of $\Sigma$.

First, we define an equivalent notion of single-peakedness in terms of vectors in
$\mathbb R^n$. Fix the standard linear ordering of the coordinate indices $1<2<\cdots<n$.
Given $x=(x_1,\ldots,x_n)\in \mathbb R^n$, for each index $m\in\{1,\ldots,n\}$,
define the strict upper level set of $x_m$ by
\[
U_x(m):=\{i\in\{1,\ldots,n\}\mid x_i>x_m\}.
\]
We say that $x$ is \emph{unimodal} with respect to this ordering if
$U_x(m)$ is an interval for every $m$. We will also define the empty set as an
interval in this case.

\begin{prop} \label{single-peaked is a subcomplex}
Let $\mathcal{A}=\{a_1,\ldots,a_n\}$ be a finite set of alternatives with the fixed linear ordering $a_1<a_2<\cdots<a_n$.
Let $\pi=(B_1|\cdots|B_k)$ be the corresponding ordered partition of a weak preference, and let $C_{\pi}$ be the corresponding cell in the type $A_{n-1}$ braid arrangement. Then the following two statements are equivalent:
\begin{itemize}
    \item[(i)] $\pi$ is single-peaked.

    \item[(ii)] every $x \in C_{\pi}$ is unimodal.
\end{itemize}
We define $\Sigma_{SP}$ to be the set containing all cells $C_{\pi}$ that correspond to single-peaked preferences $\pi \in P_{SP}$. $\Sigma_{SP}$ is a well-defined subcomplex of $\Sigma$. 
\end{prop}

\begin{proof}
By definition,
\[
C_\pi=\{x\in V \mid x_i=x_j \text{ if } a_i,a_j\in B_r;
\ x_i>x_j \text{ if } a_i\in B_r,\ a_j\in B_s,\ r<s\},
\]
From the construction of $C_\pi$, the coordinate inequalities in $C_\pi$ record exactly the strict preference relations in $\pi$. Therefore, $\pi$ is single-peaked if and only if every $x\in C_\pi$ is unimodal.

To show that $\Sigma_{SP}$ is a subcomplex of $\Sigma$, it is enough to show
that every face of a cell in $\Sigma_{SP}$ is again in $\Sigma_{SP}$.

Let $A\in \Sigma_{SP}$ and let $F\leq A$ be a face of the cell $A$. Since $F$ is obtained from the ordered partition of $A$ by repeatedly merging adjacent blocks, it is therefore sufficient to consider one such merger. Suppose that
\[
\pi_A = (B_1| \cdots |B_r|B_{r+1}| \dots |B_k)
\qquad
B_1 > \cdots > B_k,
\]
and we merge $B_r$ and $B_{r+1}$. 
\[
(B_1| \cdots |B_r \cup B_{r+1}| \dots |B_k)
\]
Then the strict upper level set $U_x(i)$ of every $x_i \in B_{r+1}$ becomes the strict upper level set $U_x(j)$ of every $x_j \in B_{r}$. Since $A \in \Sigma_{SP}$, the strict upper level set of every coordinate value is an interval. Thus, for $F$, $U_x(i) = U_x(j)$ is an interval. Because merging $B_r$ and $B_{r+1}$ does not affect the strict upper level sets of coordinate values that are not in $B_r$ and $B_{r+1}$, the resulting preference is single-peaked. Repeating this argument shows that every face $F\leq A$ lies in  $\Sigma_{\mathrm{SP}}$. Thus, $\Sigma_{SP}$ is a subcomplex of $\Sigma$.
\end{proof}

Fix an arbitrary Euclidean space $\mathbb{R}^n$, where $n \ge 3$, with the standard linear ordering of the coordinate indices $1 < 2 < \cdots < n$. For each index $p \in \{1,\dots, n\}$, we define $U_p$ as follows:
\[
U_p := \{x = (x_1, \dots, x_n) \in V \mid x_1 \leq \cdots \leq x_p \ge \cdots \ge x_n\},
\]
$U_p$ is a cone (i.e., closed under positive scalar multiplication). We call $U_p$ the \emph{peaked cone} at $p$.

\begin{prop} \label{peaked cone}
$K_p := \{A \in \Sigma \mid A \subset U_p\}$ is a subcomplex of $\Sigma_{SP}$. $U_p$ is a pointed, closed, and convex cone. 
\end{prop}

\begin{proof}
First, we will show $K_p \subset \Sigma_{SP}$. Take an arbitrary cell $A \in K_p$. We know that $A \subset U_p$, so every $x \in A$ satisfies the relation $x_1 \leq \cdots \leq x_p \ge \cdots \ge x_n$. Such an $x$ is unimodal because every strict upper-level set $U_x(m) := \{i: x_i > x_m\}$ is an interval. By Proposition \ref{single-peaked is a subcomplex}, $A \in \Sigma_{SP}$. Since we've shown that $K_p \subset \Sigma_{SP}$, we will write $K_p := \{A \in \Sigma_{SP} \mid A \subset U_p\}$ from now on. Now we show $K_p$ is a well-defined subcomplex. Let $A$ be an arbitrary cell in $K_p$. Then we know the ordered partition $\pi_{A}$ of $\{x_1 ,\dots, x_n\}$ satisfies $x_1 \leq \cdots \leq x_p \ge \cdots \ge x_n$. By the definition of the face relation, $F$ is obtained from $A$ by changing some strict inequalities into equalities. Thus, the ordered partition $\pi_{F}$ also satisfies $x_1 \leq \cdots \leq x_p \ge \cdots \ge x_n$. 

By definition, $U_p$ is the intersection of finitely many closed half-spaces $\{x_i \leq x_{i+1}\} (i < p)$ and $\{x_i \ge x_{i+1}\} (i \ge p)$, so $U_p$ is a closed and convex polyhedral cone. 
It remains to show that $U_p$ is pointed. Suppose that $x\in U_p\cap(-U_p)$. Since both $x$ and $-x$ belong to $U_p$, each of the inequalities defining $U_p$ must hold in both directions. Hence, $x_i=x_{i+1}$ for every $1\leq i<n$, so $x_1=\cdots=x_n$.
Because $x\in V$, we also have $\sum_{i=1}^n x_i=0$, which forces $x=0$. Therefore, $U_p\cap(-U_p)=\{0\}$, proving that $U_p$ is pointed.
\end{proof}

We now show that the union of all subcomplexes $K_p$ is equal to the single-peaked subcomplex.

\begin{prop} \label{SP = union Kp}
Let $K_p=\{A\in \Sigma_{SP} : A\subseteq U_p\}$. Then
\[
\Sigma_{SP}=\bigcup_{p=1}^n K_p.
\]
Equivalently,
\[
\bigcup_{A\in \Sigma_{SP}} A=\bigcup_{p=1}^n U_p.
\]
\end{prop}

\begin{proof}
Let $A$ be an arbitrary cell in $\Sigma_{SP}$. By Proposition \ref{single-peaked is a subcomplex}, every point of $A$ is unimodal. Let $\pi_A = (B_1| \cdots |B_k)$ denote the ordered partition of $\{x_1, \dots, x_n\}$ corresponding to $A$. Let $x$ be an arbitrary point in $A$ and $p$ be an index such that $x_p \in B_1$. Since
$x$ is unimodal, we must have
\[
x_1\leq \cdots \leq x_p \geq \cdots \geq x_n.
\]
 Otherwise, some strict upper level set would contain two indices but miss an index between them, contradicting unimodality. Hence, $x\in U_p$. Since $A$ is a
cell, the weak inequalities satisfied by one point of $A$ are satisfied by every
point of $A$. Thus $A\subseteq U_p$, so $A\in K_p$. Therefore $\Sigma_{SP}\subseteq \bigcup_{p=1}^n K_p$.

Conversely, since $K_p \subset \Sigma_{SP}$ for all $p \in \{1, \dots, n\}$, it follows that $\bigcup_{p=1}^n K_p \subseteq \Sigma_{SP}$.
Thus,
\[
\Sigma_{SP}=\bigcup_{p=1}^n K_p.
\]

Equivalently, since $U_p$ is the union of the cells in $K_p$ for
each $p$, taking the union of the corresponding cells gives
\[
\bigcup_{A\in\Sigma_{\mathrm{SP}}} A
=
\bigcup_{p=1}^{n} U_p.
\]
\end{proof}

From Proposition \ref{prop:preference-complex-realization}, we know there is a homeomorphism $\phi: |\Sigma| \to S(V)$. Now we will study $\phi$ when restricted to $|\Sigma_{SP}|$ and each $|K_p|$.

\begin{prop} \label{|Sigma_{SP}| Homeo}
Given a finite set of alternatives $\mathcal{A} = \{a_1, \dots, a_n\}$, let $\Sigma_{SP}$ be the single-peaked subcomplex of the preference complex $\Sigma$. From Proposition \ref{prop:preference-complex-realization}, there exists a homeomorphism $\phi: |\Sigma| \to S(V)$. The following restrictions of $\phi$ give
homeomorphisms
\[
\phi|_{|\Sigma_{SP}|}:
|\Sigma_{SP}|
\longrightarrow
\bigcup_{p=1}^{n}\bigl(U_p\cap S(V)\bigr)
\]
and, for every $p\in\{1,\ldots,n\}$,
\[
\phi|_{|K_p|}:
|K_p|
\longrightarrow
U_p\cap S(V).
\]
\end{prop}

\begin{proof}
By definition,
\[
|\Sigma_{SP}| = \bigcup_{\varnothing \neq A \in \Sigma_{SP}}|A|,
\]
where $A$ ranges over the nonempty simplices in $\Sigma_{SP}$, and $|A|$ denotes the corresponding open geometric simplices. By Proposition \ref{prop:preference-complex-realization},
\[
\phi|_{|A|}:|A|\longrightarrow A\cap S(V).
\]
Since the open simplices $|A|$ are pairwise disjoint, 
\[
\phi(|\Sigma_{SP}|) = \bigcup_{\varnothing \neq A \in \Sigma_{SP}}\phi(|A|) = \bigcup_{\varnothing \neq A \in \Sigma_{SP}} A \cap S(V).
\]
By Proposition \ref{SP = union Kp}, we know that 
\[
\bigcup_{A\in \Sigma_{SP}}A=\bigcup_{p=1}^n U_p.
\]
Thus, we have
\[
\phi(|\Sigma_{SP}|) = \bigcup_{p=1}^n U_p \cap S(V).
\]
Because $\phi: |\Sigma| \to S(V)$ is a homeomorphism and $|\Sigma_{SP}| \subset |\Sigma|$, we know that $\phi_{|\Sigma_{SP}|}$ is a homeomorphism onto its image. 

For each $p \in \{1, \dots, n\}$, applying the same argument to $K_p$, we obtain
\[
\phi(|K_p|)=U_p\cap S(V).
\]
Therefore, the restriction
\[
\phi|_{|K_p|}:|K_p|\longrightarrow U_p\cap S(V)
\]
is a homeomorphism.
\end{proof}

We now show a nice property of intersections of $U_p \cap S(V)$. This property will be crucial to the proof that $|\Sigma_{SP}|$ is contractible once we introduce the nerve complex.

\begin{prop} \label{peaked cones are contractible}
Let $I\subseteq \{1,\ldots,n\}$ be nonempty, and define
\[
U_I:=\bigcap_{p\in I}U_p
=\Bigl\{x\in V \;\Bigm|\;
x_1\le\cdots\le x_{\min I}=\cdots=x_{\max I}\ge\cdots\ge x_n\Bigr\}.
\]
Then $U_I$ is a pointed, closed, and convex cone. If $U_I\neq \{0\}$, then
\[
U_I\cap S(V)
\]
is contractible.
\end{prop}

\begin{proof}
Take an arbitrary $I \subseteq \{1, \dots , n\}$. For every $p \in I$, by Proposition \ref{peaked cone}, $U_p$ is a pointed, closed, and convex cone. Since $U_I$ is a finite intersection of peaked cones $U_p$, $U_I$ is a closed and convex cone. Furthermore, for any index $p \in I$, 
\[
U_I \cap -U_I \subseteq U_p \cap -U_p = \{0\}.
\]
Thus, $U_I$ is also pointed.

Assume $U_I \neq \{0\}$. Since $U_I$ is a cone, we can rescale any nonzero element in it to a unit vector that is still in $U_I$. Fix such a unit vector $u_0 \in U_I \cap S(V)$. Define a function $H: U_I \cap S(V) \times [0,1] \to U_I \cap S(V)$ by
\[
H(x, t) = \frac{(1-t)x + tu_0}{\|(1-t)x + tu_0\|}.
\]

Notice that $(1-t)x + tu_0 \neq 0$. This is because, by convexity of $U_I$, $(1-t)x + tu_0 \in U_I$. If $(1-t)x + tu_0 = 0$, then $x = -\frac{t}{1-t} u_0 \in U_I \cap -U_I = \{0\}$, contradicting the fact that $\|x\| = 1$.

The map $H$ is well-defined because $U_I$ is a convex cone, so scaling $(1-t)x + tu_0 \in U_I$ down by its positive norm keeps the vector in $U_I$, making $H(x,t) \in U_I \cap S(V)$. By construction, the map $H$ is continuous. Since $H(x,0) = x$, $H(x,1) = u_0$, and $H(u_0 ,t) = u_0$ for all $t$, $H$ is a deformation retraction of $U_I \cap S(V)$ onto the point $u_0$. Thus, $U_I \cap S(V)$ is contractible. 
\end{proof}

In Proposition \ref{peaked cones are contractible}, we showed that if $U_I \neq \{0\}$, $U_I \cap S(V)$ is contractible. We will study the condition on the index set $I$ under which $U_I \neq \{0\}$.

\begin{prop} \label{the only empty is {1,n}}
Let $I \subseteq \{1, \dots, n\}$ be nonempty. Then $\{1, n\} \subseteq I$ if and only if $U_I \cap S(V) = \varnothing$.
\end{prop}

\begin{proof}
Suppose that $\{1, n\} \subseteq I$. Then we know that min$I = 1$ and max$I = n$. Thus,
\[
U_I = \{x \in V \mid x_1 = \cdots = x_n\}.
\]
The only point in $V$ that satisfies $x_1 = \cdots = x_n$ is $0$. However, $0 \notin S(V)$. Thus, $U_I \cap S(V) = \varnothing$.

Conversely, if $\{1, n\} \not\subseteq I$, we can construct a point $y \in U_I$ such that $y \neq 0$. When $p = \operatorname{min}I > 1$, we define $y \in U_I \cap V$ as
\[
y_i = \begin{cases} 1, & i \ge p, \\ -\frac{n-p+1}{p-1}, & i < p. \end{cases}
\]
Similarly, when $\operatorname{min}I = 1$, we know $q = \operatorname{max} I < n$. We define $y \in U_I \cap V$ as
\[
y_i = \begin{cases} 1, & i \leq q, \\ -\frac{q}{n-q}, & i > q. \end{cases}
\]
Since $U_I$ is a cone, we can normalize $y$ to find an element in $U_I \cap S(V)$. Thus, when $\{1, n\} \not\subseteq I$, $U_I \cap S(V)$ is nonempty.
\end{proof}

Now we introduce a topological tool that we will use for our final proof that $|\Sigma_{SP}|$ is contractible.

\begin{defn}
Let $\mathcal{B} = \{B_1, \dots, B_n\}$ be a family of sets. The \emph{nerve complex} of $\mathcal{B}$, denoted by $N(\mathcal{B})$, is the abstract simplicial complex with the vertex set $\{1, \dots, n\}$ and with simplices given by
\[
N(\mathcal{B}) := \{\sigma \subseteq \{1, \dots, n\} \mid \bigcap_{i \in \sigma} B_i \neq \varnothing\}.
\]
\end{defn}

\begin{thm}  [Borsuk's Nerve Theorem \cite{matousek}] \label{nerve's theorem}
Let $K_1, \dots, K_n$ be subcomplexes of a finite simplicial complex $K$ that cover $K$, and let $B_i := |K_i|$. Suppose that the intersection $\bigcap_{i 
\in I} B_i$ is either empty or contractible for each nonempty $I \subseteq \{1, \dots, n\}$. Then the nerve complex $N(\{B_1, \dots, B_n\})$ is homotopy equivalent to $|K|$.
\end{thm}

Now we prove $|\Sigma_{SP}|$ is contractible using Theorem \ref{nerve's theorem}.

\begin{prop} \label{N(U) is homotopy equivalent to |Sigma_sp|}
Let $\mathcal{U} := \{\, U_p \cap S(V) \mid 1 \le p \le n \,\}$. Then $N(\mathcal{U})$ is homotopy equivalent to $|\Sigma_{SP}|$.
\end{prop}

\begin{proof}
By Proposition \ref{|Sigma_{SP}| Homeo}, $\phi$ restricts to a homeomorphism
$|K_p| \to U_p \cap S(V)$ for each $p \in \{1, \dots, n\}$. Since $\phi$ is a homeomorphism from $|\Sigma|$ to $S(V)$ (Proposition \ref{prop:preference-complex-realization}), it is injective. Thus, $\phi$ commutes with
intersections: for every nonempty $I \subseteq \{1,\dots,n\}$,
\[
  \phi \; \!\Big(\bigcap_{p\in I}|K_p|\Big)
  = \bigcap_{p\in I} \phi(|K_p|)
  = \bigcap_{p\in I}\big(U_p \cap S(V)\big)
  = \Big(\bigcap_{p\in I} U_p\Big)\cap S(V)
  = U_I \cap S(V).
\]
By Propositions \ref{peaked cones are contractible} and \ref{the only empty is {1,n}}, $U_I \cap S(V)$ is contractible when $\{1,n\}\not\subseteq I$ and empty when $\{1,n\}\subseteq I$. Because $\phi$ is a homeomorphism, $\bigcap_{p\in I}|K_p|$ is likewise empty or contractible.
From Proposition \ref{SP = union Kp}, the $K_p$ are subcomplexes of the finite simplicial complex $\Sigma_{SP} = \bigcup_{p=1}^n K_p$, so Theorem \ref{nerve's theorem} applies and shows that $N(\{|K_p|\})$ is homotopy equivalent to $|\Sigma_{SP}|$. Finally, since $\phi$ is a bijection, $\bigcap_{p\in I}|K_p|\neq\varnothing$ if and only if $\bigcap_{p\in I}(U_p\cap S(V))\neq\varnothing$, so the two families have the same nerve complex, $N(\mathcal{U}) = N(\{|K_p|\})$. Hence, $N(\mathcal{U})$ is homotopy equivalent to $|\Sigma_{SP}|$.
\end{proof}

By Proposition \ref{N(U) is homotopy equivalent to |Sigma_sp|}, the contractibility of 
$ \lvert \Sigma_{\mathrm{SP}} \rvert $ reduces to that of $N(\mathcal{U})$.
We now show the structure of $N(\mathcal{U})$ explicitly. Before doing so,
we introduce a new notation. We write $\Delta(1,\ldots,n)$ for the abstract
simplicial complex consisting of the $(n-1)$-simplex with vertex set
$\{1,\ldots,n\}$, together with all its faces.

\begin{prop} \label{N(U) is two simplices}
The nerve complex $N(\mathcal{U})$ is given by
\[
N(\mathcal{U})
=
\Delta(1,\ldots,n-1)\cup\Delta(2,\ldots,n).
\]
The two simplices intersect in the common subcomplex
$\Delta(2,\ldots,n-1)$.
\end{prop}

\begin{proof}
Fix an arbitrary simplex $\sigma \in N(\mathcal{U})$. We know that $\sigma \subset \{1, \dots, n\}$. By definition of the nerve complex, $\sigma$ is a simplex in $N(\mathcal{U})$ if and only if
\[
U_{\sigma} \cap S(V) \neq \varnothing.
\]
From Proposition \ref{the only empty is {1,n}}, $\{1, n\} \not\subseteq \sigma$, meaning that $\sigma$ is in $\Delta(1, \dots, n-1) \cup \Delta(2, \dots, n)$. Because $\sigma$ is an arbitrary simplex in $N(\mathcal{U})$, $N(\mathcal{U}) \subseteq \Delta(1, \dots, n-1) \cup \Delta(2, \dots, n)$. Before showing the other direction, we will make a quick remark on $\sigma \in N(\mathcal{U})$. If neither $1$ nor $n$ is in $\sigma$, then $\sigma$ is in both $\Delta(1, \dots, n-1)$ and $\Delta(2, \dots, n)$. This is because these two $(n-2)$-simplices share a common $(n-3)$-simplex $\Delta(2, \dots, n-1)$. On the other hand, if $\sigma$ contains either $1$ or $n$, then $\sigma$ is in either $\Delta(1, \dots, n-1)$ or $\Delta(2, \dots, n)$.

Conversely, fix an arbitrary simplex $\delta \in \Delta(1, \dots, n-1) \cup \Delta(2, \dots, n)$. By our construction, $\{1, n\} \not\subseteq \delta$. Applying Proposition \ref{the only empty is {1,n}}, we get that
\[
U_{\delta} \cap S(V) \neq \varnothing.
\]
Thus, by definition of the nerve complex, $\delta \in N(\mathcal{U})$. Since $\delta$ is arbitrary, $\Delta(1, \dots, n-1) \cup \Delta(2, \dots, n) \subseteq N(\mathcal{U})$. We conclude that $N(\mathcal{U}) = \Delta(1, \dots, n-1) \cup \Delta(2, \dots, n)$.
\end{proof}

Now we show that $N(\mathcal U)$ is contractible. To prove this, we will show $N(\mathcal U)$ is a cone.

\begin{prop} \label{N(U) is contractible}
The nerve complex $N(\mathcal U)$ is a cone. Hence, $N(\mathcal U)$ is contractible.
\end{prop}

\begin{proof}
From Proposition \ref{N(U) is two simplices}, $N(\mathcal{U}) = \Delta(1, \dots, n-1) \cup \Delta(2, \dots, n)$. Let $v$ be any vertex of $\Delta(2, \dots, n-1)$. Fix an arbitrary $\sigma \in N(\mathcal{U})$. We know that $\{1, n\} \not\subseteq \sigma$. Since $v$ is a vertex in $\Delta(2, \dots, n-1)$, $v \notin \{1, n\}$. Hence, $\{1, n\} \not\subseteq \sigma \cup \{v\}$, meaning that $\sigma \cup \{v\} \in N(\mathcal{U})$. Since $\sigma$ is arbitrary, this result holds for every simplex in $N(\mathcal{U})$, proving $N(\mathcal{U})$ is a cone. Since all cones are contractible, $N(\mathcal{U})$ is contractible. 
\end{proof}

Propositions \ref{N(U) is homotopy equivalent to |Sigma_sp|} and \ref{N(U) is contractible} together prove Theorem \ref{single-peakedness is contractible}.
\end{appendices}

\end{document}